\documentclass[a4paper, reqno]{amsart}
\usepackage[a4paper,top=4cm,bottom=4cm,left=3cm,right=3cm]{geometry}

\usepackage[T1]{fontenc}
\usepackage[utf8]{inputenc}

\usepackage{amsmath,amsthm,amssymb,mathtools}
\usepackage{mathabx}
\usepackage{microtype}
\usepackage{enumitem}
\usepackage{comment}
\usepackage{mathrsfs}
\usepackage[dvipsnames]{xcolor}
\usepackage[linktocpage]{hyperref}
\hypersetup{
  colorlinks=true,
  linkcolor=Maroon,
  citecolor=Maroon,
  urlcolor=Maroon
}
\usepackage[capitalise,nameinlink]{cleveref}
\usepackage{biblatex}
\numberwithin{equation}{section}

\theoremstyle{plain}
\newtheorem{theorem}{Theorem}[section]
\newtheorem{proposition}[theorem]{Proposition}
\newtheorem{lemma}[theorem]{Lemma}
\newtheorem{corollary}[theorem]{Corollary}

\theoremstyle{definition}
\newtheorem{definition}[theorem]{Definition}

\theoremstyle{remark}
\newtheorem{remark}[theorem]{Remark}

\newcommand{\R}{\mathbb{R}}
\newcommand{\dd}{\mathrm{d}}
\newcommand{\Id}{\mathrm{Id}}

\newcommand{\tr}{\mathrm{tr}}

\newcommand{\lam}{\lambda}
\newcommand{\lamo}{\lambda_0}
\newcommand{\lamone}{\lambda_1}

\usepackage{indentfirst}
\title[Obstructions to global visibility of singularities]{Obstructions to global visibility of singularities in asymptotically flat spacetimes}

\author{Kanabar Jay}
\address{Department of Mathematics, MG Science Institute, Gujarat University, Ahmedabad 380009, India}
\email{\textcolor{Maroon}{kanabarjaymgsi@gmail.com}}

\author{Kharanshu N. Solanki}
\address{Fakultät für Mathematik, Universität Wien,  Oskar-Morgenstern-Platz 1, 1090 Wien, Austria}
\email{\textcolor{Maroon}{kharanshu.nilesh.solanki@univie.ac.at}}

\author{Pankaj S. Joshi}
\address{International Centre for Space and Cosmology, Ahmedabad University, Ahmedabad 380009, India}
\email{\textcolor{Maroon}{psjcosmos@gmail.com}}

\thanks{This research was funded in part by the Austrian Science Fund (FWF) [Grant DOI \href{https://www.fwf.ac.at/en/research-radar/10.55776/STA32}{10.55776/STA32}]. For open access purposes, the authors have applied a CC BY public copyright license to any author-accepted manuscript version arising from this submission.}

\newcommand{\printtitledate}[1]{%
  \vspace*{-0.6\baselineskip}%
  \begin{center}\normalfont\small #1\end{center}%
  \vspace*{-0.3\baselineskip}%
}

\begin{document}
\maketitle
\printtitledate{\today}
\begin{abstract}
Consider an $(N+1)$-dimensional asymptotically flat spacetime and a future-directed, affinely parametrized outgoing null generator $\gamma$ of an achronal boundary $\partial J^+(S_\varepsilon)$, where $\{S_\varepsilon\}$ is a nested family of smooth compact codimension $2$ surfaces approaching a singular boundary set \(S\) in the past. In the twist-free case and under the null energy condition, the Raychaudhuri equation on the $m:=N-1$ dimensional screen bundle reads,
$$
\theta'=-\frac1m\theta^2-\|\sigma\|^2-\mathrm{Ric}(k,k),
$$
where $k$ is the tangent to $\gamma$. This equation linearizes, via the rescaling $u:=A^{1/m}$ with $A \coloneqq|\det D|$ the Jacobi-map \(m\)-volume, to the Sturm-type ODE
$$
u''+\frac1m f\,u=0,\qquad f:=\|\sigma\|^2+\mathrm{Ric}(k,k)\ge 0.
$$
We develop two purely generator-wise criteria forcing a first zero of $u$: (i) an exact Volterra identity combined with concavity leads to a barrier-weighted integral inequality, and (ii) Sturm comparison and a Prüfer-angle estimate yields failure of disconjugacy whenever $\int_c^d \sqrt{f/m}\,d\lambda>\pi$ on a subinterval. We prove that $u(\lambda_\ast)=0$ is equivalent to the existence of a focal (conjugate) point and implies $\theta= m u'/u\to-\infty$ at $\lambda_\ast$. Using the standard structure of achronal boundaries, this yields a geodesic-wise obstruction: if every generator that could reach $\mathscr I^+$ satisfies one of the above conditions in the regular spacetime region, then $J^+(S_\varepsilon)\cap \mathscr I^+=\emptyset$, and hence \(S\) is not globally visible. As an application, we illustrate one of these criteria in the Einstein-massless scalar field collapse model of Christodoulou.
\end{abstract}
\setcounter{tocdepth}{2}

\setcounter{tocdepth}{1}
\tableofcontents

\section{Introduction}

The classical singularity theorems of general relativity show that under appropriate energy and causality conditions, spacetime geodesics are incomplete \cite{Hawking:1973uf}. What these theorems do \emph{not} decide by themselves is whether the singular boundary (associated with this geodesic incompleteness) is hidden behind a trapped region or is (at least locally) visible. The problem of finding suitable conditions that lead to the formation of locally naked singularities is still a vastly open, and typically quite challenging pursuit. Here we address a related problem of finding obstructions for a locally visible singularity in a certain class of spacetimes to be globally visible, in the sense that any null geodesic issuing from the locally naked singularity cannot escape to future null infinity. Note that in particular such obstructions do not claim anything about cosmic censorship \cite{ShlapentokhRothman2025WCC} directly, since one already assumes the existence of a past-singular boundary of the spacetime under consideration. However these obstructions and the methods used to derive them may be of independent interest in order to obtain a better understanding of certain aspects of naked singularities. In particular the conditions presented here will demonstrate how the behavior of curvature and shear along null geodesics emanating from such singularities affect their global visibility. The methods rely heavily on various notions from Lorentzian causality theory, for which the unfamiliar reader is referred to \cite{Minguzzi:2019mbe}. Note that such methods have already been employed in the past by Tipler \cite{TIPLER1978165,Tipler:1977zza,Tipler:1977zzb}, Mingarelli \cite{minga}, and Hawking and Penrose \cite{Hawking:1970zqf} in various different contexts. 

\subsection{Setup and guiding idea}
Fix a future-directed, affinely parametrized null geodesic $\gamma:I\to M$, and choose an auxiliary null vector field $n\in \Gamma(TM)$ with $g(k,n)=-1$, where $k$ is the tangent to $\gamma$. The \emph{screen bundle} of $M$ is the $m\coloneqq N-1$-dimensional vector bundle over $M$ with the underlying set $S_\lambda:=\{X:\ g(X,k)=g(X,n)=0\}$. The deformation of a thin null congruence around $\gamma$ is encoded by the \emph{optical endomorphism} $B:S_\lambda\to S_\lambda$, defined by $(J^\perp)'=B(J^\perp)$ for screen-projected Jacobi fields $J^\perp$.  It splits as,
$$\displaystyle B=\frac{1}{m}\theta\,\Id+\sigma+\omega,$$
where $\theta$ is the expansion, $\sigma$ the trace-free shear, and $\omega$ the twist. Throughout we work in the twist-free case ($\omega\equiv 0$) and impose the null energy condition (NEC), i.e., $\mathrm{Ric} (k,k)=T(k,k)\ge 0$.

Let $\Pi:TM\to S_{\lambda}$ be the orthogonal projection to the screen-bundle. Then the trace of the optical evolution identity $\nabla_{k}B=-B^2-\Pi\circ R(\cdot,k)k$ yields the Raychaudhuri equation
\begin{equation}\label{eq:intro-ray}
\theta'=-\frac{1}{m}\theta^2-\|\sigma\|^2-\mathrm{Ric}(k,k), 
\end{equation}
where $(\cdot)'$ denotes differentiation with respect to the affine parameter $\lambda$ (we shall use this notation throughout the paper). Let $D(\lam):\mathcal \mathcal S_{\lamo}\to \mathcal S_\lam$ be the \emph{Jacobi map} that sends initial screen vectors at $\lamo$ to the corresponding Jacobi fields at $\lam$. To extract a \emph{linear} equation from the Raychaudhuri equation, we rescale the congruence’s $m$-volume $\mathcal A(\lambda):=|\det D(\lambda)|$ (with $D' = B D$, $D(\lambda_0)=\Id$) by,
$$\displaystyle u(\lambda):=\mathcal A(\lambda)^{1/m}\quad\Longrightarrow\quad \theta=m\,\frac{u'}{u}.$$
Differentiating with respect to $\lambda$ and using \eqref{eq:intro-ray} cancels the quadratic term and yields the scalar ODE
\begin{equation}\label{eq:intro-ode}
u''+\frac{1}{m}f\,u=0,\qquad f:=\|\sigma\|^2+T(k,k).
\end{equation}
Equation \eqref{eq:intro-ode} is the backbone of our analysis. It reduces the geometric focusing problem to a second-order Sturm-type equation along \emph{one} null geodesic (which is part of a congruence).

\subsection{Zero of \(u\), conjugate points, and pre-horizon placement}
The existence of a first zero of $u$ has two immediate consequences. First, $u(\lambda_*)=0$ if and only if $\det D(\lambda_*)=0$, i.e. there is a nontrivial screen Jacobi field $J$ with $J(\lambda_*)=0$ (a conjugate point). Second, $\theta=m\,u'/u\to -\infty$ as $\lambda\to\lambda_*$, which implies that the focusing is genuine and not a coordinate artifact. If in addition the spacetime is asymptotically flat and its future null infinity satisfies certain properties (as made precise in Theorem \ref{thm:global_visibility_failure}), and moreover the null geodesic (and the congruence thereof) is past-incomplete, i.e., it terminates at a singularity in the past, and the first zero occurs at $\lambda_*<\lambda_{\mathrm{AH}}$ (where $\lambda_{\mathrm{AH}}$ is the first parameter where the outgoing expansion $\theta_+$ of a suitable congruence vanishes), and moreover the segment up to $\lambda_{\mathrm{AH}}$ is regular and untrapped, then we shall show  the singularity is \emph{not globally visible} along $\gamma$. 

\subsection{Two possible obstructions to global visibility}
We provide two different routes to force a pre-horizon zero of $u$:
\begin{enumerate}[leftmargin=2em]
\item Barrier route: On any interval where $f\ge 0$ and the affine-linear \emph{barrier} $b(\lambda)=u(\lambda_0)+u'(\lambda_0)(\lambda-\lambda_0)$ stays positive, a weighted inequality of the form
$$\frac{1}{m}\int (\lambda_1-s)\,f(s)\,b(s)\,\dd s>u(\lambda_0)+u'(\lambda_0)(\lambda_1-\lambda_0)$$
forces $u(\lambda_1)<0$, thereby guaranteeing a zero of $u$ in $(\lambda_0,\lambda_1)$.
\item Sturm/Prüfer route: Writing $q:=f/m$, if $\int_{c}^{d}\sqrt{q(\lambda)}\,\dd\lambda>\pi$ on some subinterval $(c,d)$ contained in the regular region, then the solution of \eqref{eq:intro-ode} has a zero of $u$ in $(c,d)$.
\end{enumerate}
Both tests require a condition on $f$ along a single chosen geodesic. Recall that all of this holds under the twist-free evolution along $\gamma$ along with the NEC. The auxiliary null vector field $n$ only determines the screen and eventually drops out of all invariants. Moreover, affine reparametrizations leave the criteria unchanged.

\section{Geometric framework and preliminaries}\label{sec:framework}

Let $(M,g)$ be a time-oriented Lorentzian manifold. Fix a future-directed null geodesic congruence with tangent $k$, affinely parametrized by $\lam$ so that $\nabla_k k=0$ along each geodesic of the congruence. Choose an auxiliary future-directed null vector field $n$ with $g(k,n)=-1$. The \emph{screen bundle} at parameter $\lam$ is the $m$-dimensional vector bundle with the underlying set,
$$\mathcal S_\lam:=\{X\in T_{\gamma(\lam)}M:\ g(X,k)=g(X,n)=0\},\qquad m:=N-1,$$
with inner product $h:=g|_{\mathcal S_\lam}$. Let $\Pi:T M\to \mathcal S_\lam$ be the \emph{orthogonal projector}.

\begin{definition}[Optical map and scalars]\label{def:optical}
If $J$ is a Jacobi field along a geodesic, its screen projection $J^\perp:=\Pi\circ J$ satisfies $(J^\perp)'=B(J^\perp)$, defining the \emph{optical endomorphism} $B:\mathcal S_\lam\to \mathcal S_\lam$. One may decompose
$$B=\frac{1}{m}\theta\,\Id+\sigma+\omega,\qquad \theta:=\tr B,\quad \tr\sigma=0,\quad \omega^\top=-\omega.$$
We assume that the congruence is \emph{twist-free} ($\omega\equiv 0$) throughout. Define $\|\sigma\|^2:=\tr(\sigma^\top\sigma)\ge 0$.
\end{definition}

\begin{remark}[Energy condition]
Under the NEC, $T(X,X)\ge 0$ for every null $X$. Einstein’s equations in $(N+1)$ dimensions yield $\mathrm{Ric}(k,k)=T(k,k)\ge 0$.
\end{remark}

\subsection{The Raychaudhuri equation}\label{sec:ray}

For the sake of completeness, we first establish the optical evolution identity, then derive the Raychaudhuri equation by taking the screen bundle trace.

\begin{lemma}\label{lem:optical-evolution}
Let $R(U,V)W:=\nabla_U\nabla_VW-\nabla_V\nabla_UW-\nabla_{[U,V]}W$ be the Riemann curvature operator for $U,V,W\in\Gamma(TM)$. Along each geodesic $\gamma$ with tangent $k$, one has,
\begin{equation}\label{optical-identity}
    \nabla_k B=-B^2-\Pi\circ R(\,\cdot\,,k)k\big|_{S}.
\end{equation}
\end{lemma}
\begin{proof}
Fix a geodesic $\gamma:I\to M$ and a screen Jacobi field $J$ along it, i.e. $J(\lam)\in \mathcal S_\lam$ and $g(J(\lam),k)=g(J(\lam),n)=0$ for all $\lam\in I$. The Jacobi equation is given by $J''=R(k,J)k$. Projecting both sides with $\Pi$ yields $\Pi\circ J''=\Pi\circ R(k,J)k$. Since $J^\perp:=\Pi \circ J$ and $\Pi$ is parallel along $k$ on $S$ (we identify screen fibers by parallel transport), the left hand side equals $(J^\perp)''$. By definition \((J^\perp)'=B(J^\perp)\). It follows that,
$$(J^\perp)''=B'(J^\perp)+B\big((J^\perp)'\big)=B'(J^\perp)+B^2(J^\perp).$$
Hence for every screen vector field of the form $J^\perp$ we have
$$B'(J^\perp)+B^2(J^\perp)=\Pi\,R(k,J)k.$$
As the vector fields $J^\perp$ span $\mathcal S_\lam$ (by choosing initial data for two linearly independent screen Jacobi fields and propagating), we can identify operators acting on $S$ to get
$$\nabla_k B=-B^2-\Pi\circ R(\,\cdot\,,k)k\big|_{S}.$$
This completes the proof.

\end{proof}

\begin{proposition}\label{prop:ray}
With $m=N-1$ and $\omega=0$, one has,
\begin{equation}\label{ray}
    \theta'=-\frac{1}{m}\theta^2-\|\sigma\|^2-\mathrm{Ric}(k,k).
\end{equation}
\end{proposition}
\begin{proof}
Take the screen bundle trace of equation \ref{optical-identity} gives,
$$\tr(\nabla_k B)=-\tr(B^2)-\tr\big(\Pi\circ R(\,\cdot\,,k)k\big).$$
Since $\tr(\nabla_k B)=\nabla_k(\tr B)=\theta'$, it remains to evaluate $\tr(B^2)$ and the curvature trace. Write $B=(\theta/m)\,\Id+\sigma$. It follows that,
\begin{align*}
B^2&=\Big(\frac{1}{m}\theta\,\Id+\sigma\Big)^2
    =\frac{\theta^2}{m^2}\Id+\frac{2\theta}{m}\sigma+\sigma^2,\\
\tr(B^2)&=\frac{\theta^2}{m^2}\tr(\Id)+\frac{2\theta}{m}\tr(\sigma)+\tr(\sigma^2)
=\frac{\theta^2}{m}+\|\sigma\|^2,
\end{align*}
Here we have used $\tr(\Id)=m$ and $\tr(\sigma)=0$. For the curvature, orthogonal projection does not affect the trace on $S$, and hence,
$$\tr\big(\Pi\circ R(\,\cdot\,,k)k\big)=\tr\big(R(\,\cdot\,,k)k|_{S}\big)=\mathrm{Ric}(k,k).$$
Therefore $\theta'=-\theta^2/m-\|\sigma\|^2-\mathrm{Ric}(k,k)$ as claimed.

\end{proof}

\subsection{Jacobi map, rescaling, and linearization}\label{sec:linearization}

\begin{definition}[Jacobi map and beam $m$-volume]\label{def:Jacobi}
Fix $\lamo$ and identify $\mathcal S_{\lamo}$ with $\mathcal S_\lam$ by parallel transport along $k$. The \emph{Jacobi map} $D(\lam):\mathcal S_{\lamo}\to \mathcal S_\lam$ sends initial screen vectors to the corresponding Jacobi field values at $\lam$, and satisfies $D(\lamo)=\Id$, $D'=B D$. Define the unsigned \emph{$m$-volume density} and \emph{beam scale}
$$\mathcal A(\lam):=|\det D(\lam)|,\qquad u(\lam):=\mathcal A(\lam)^{1/m}.$$
\end{definition}

\begin{lemma}\label{lem:area}
For $u>0$, one has,
$$\theta=\frac{1}{\mathcal A}\frac{\dd\mathcal A}{\dd\lam}=m\,\frac{u'}{u}.$$
\end{lemma}
\begin{proof}
Liouville’s formula for matrix-valued ODEs gives
$$\frac{\dd}{\dd\lam}\det D=\tr(B)\,\det D=\theta\,\det D.$$
Taking the absolute value and using that $\mathcal A=|\det D|$, we obtain $\mathcal A'=\theta\mathcal A$ wherever $\det D\ne 0$ (hence $\mathcal A>0$). Since $\mathcal A=u^m$, differentiating with respect to $\lambda$ gives,
$$\frac{\mathcal A'}{\mathcal A}=\frac{(u^m)'}{u^m}=m\,\frac{u'}{u}.$$
Combining the two expressions yields the required identity.

\end{proof}

\begin{proposition}\label{prop:linear}
On any interval where $u>0$, one has,
$$u''+\frac{1}{m}\,f\,u=0,\qquad f:=\|\sigma\|^2+T(k,k)\ (\ge 0\ \text{under NEC}).$$
\end{proposition}
\begin{proof}
By \cref{lem:area}, $\theta=m\,u'/u$. Differentiating with respect to $\lambda$ gives,
$$ \theta' = m\left(\frac{u''}{u}-\frac{(u')^2}{u^2}\right). $$
Inserting into the Raychaudhuri equation \ref{ray} gives,
$$m\left(\frac{u''}{u}-\frac{(u')^2}{u^2}\right)
= -\frac{1}{m}\theta^2-\|\sigma\|^2-T(k,k)
= -\frac{1}{m}\,m^2\frac{(u')^2}{u^2}-\|\sigma\|^2-T(k,k).
$$
The terms in $(u'/u)^2$ cancel exactly, leaving
$$m\,\frac{u''}{u}= -\big(\|\sigma\|^2+T(k,k)\big).$$
Multiplying by $u$ gives $u''+(f/m)\ u=0$ with $f:=\|\sigma\|^2+T(k,k)\ge 0$ by NEC.

\end{proof}

\section{The barrier route to conjugacy}\label{sec:volterra}

\begin{lemma}\label{lem:volterra}
Let $u$ solve $u''+(f/m)\,u=0$ with $f\in C^0([\lamo,\lamone])$. Then for any initial data $u(\lamo)=u_0>0$, $u'(\lamo)=u_0'$, one has,
\begin{equation}\label{barrier}
    u(\lamone)=u_0+u_0'(\lamone-\lamo)-\frac{1}{m}\int_{\lamo}^{\lamone}(\lamone-s)\,f(s)\,u(s)\,\dd s.
\end{equation}
\end{lemma}
\begin{proof}
Integrating the ODE from $\lamo$ to $t$ gives,
$$u'(t)-u_0'=-\frac{1}{m}\int_{\lamo}^{t} f(s)u(s)\,\dd s.$$
Integrating in $t$ from $\lamo$ to $\lamone$ then gives,
$$u(\lamone)-u_0 = u_0'(\lamone-\lamo) - \frac{1}{m}\int_{\lamo}^{\lamone}\left(\int_{\lamo}^{t} f(s)u(s)\,\dd s\right)\dd t.$$
The integration domain is the triangle $\{(\!s,t): \lamo\le s\le t\le \lamone\}$. Swapping the order (Fubini/Tonelli applies since the integrand is continuous) gives,
$$\int_{\lamo}^{\lamone}\left(\int_{t=\;s}^{\lamone}\dd t\right) f(s)u(s)\,\dd s =\int_{\lamo}^{\lamone}(\lamone-s)\,f(s)u(s)\,\dd s.$$
Substituting this into the previous equation yields the claimed identity.

\end{proof}

\begin{lemma}\label{lem:concavity}
If $f\ge 0$ and $u\ge 0$ on an interval $I$, then $u''=-(1/m)fu\le 0$ on $I$. In particular, if $u>0$ on $[\lamo,\lamone]$, then $u$ is concave on that interval.
\end{lemma}
\begin{proof}
Immediate from the ODE and the nonnegativity of $f$ and $u$.

\end{proof}

\begin{theorem}\label{thm:barrier}
Let $b(s):=u_0+u_0'(s-\lamo)$ and assume $f\ge 0$ on $[\lamo,\lamone]$, $u_0>0$, $u_0'<0$, and $b>0$ on $[\lamo,\lamone]$. Then,
\begin{enumerate}[label=\textup{(\alph*)}]
\item $u\ge b$ on $[\lamo,\lamone]$, and,
\item if
$$\frac{1}{m}\int_{\lamo}^{\lamone}(\lamone-s)\,f(s)\,b(s)\,\dd s\ >\ u_0+u_0'(\lamone-\lamo),$$
then $u(\lamone)<0$. Hence there exists $\lam_*\in(\lamo,\lamone)$ with $u(\lam_*)=0$.
\end{enumerate}
\end{theorem}
\begin{proof}
(a) Define $w:=u-b$. Then $w(\lamo)=w'(\lamo)=0$. On any subinterval where $u\ge 0$, we have
$$w''=u''-b''=-\frac{1}{m}f u - 0 \le 0,$$
so $w$ is concave. If $w$ becomes negative first at $\hat\lam\in(\lamo,\lamone]$, by concavity one has $w(\hat\lam)=\min w<0$ and $w'(\hat\lam)=0$, but then the graph of $w$ would lie strictly below the tangent at $\lamo$ (which is identically zero) even for $\lam$ near $\lamo$, contradicting $w(\lamo)=w'(\lamo)=0$. Hence $w\ge 0$ and $u\ge b$.

(b) Inserting $u\ge b$ into the inequality \ref{barrier} gives,
$$u(\lamone)\le u_0+u_0'(\lamone-\lamo)-\frac{1}{m}\int_{\lamo}^{\lamone}(\lamone-s)\,f(s)\,b(s)\,\dd s.$$
If the right-hand side is negative, then $u(\lamone)<0$. Since $u(\lamo)=u_0>0$ and $u$ is continuous, the intermediate value theorem guarantees a zero in $(\lamo,\lamone)$.

\end{proof}

\section{The Sturm route to conjugacy}\label{sec:sturm}

For this section we write the Raychaudhuri equation as $y''+q(\lam)y=0$ with $q(\lam):=f(\lam)/m\ge 0.$

\begin{lemma}\label{lem:lagrange}
For $j \in\{0,1\}$, if $y_j''+q_j y_j=0$ with continuous $q_j$, then
$$\big(y_0 y_1'-y_0' y_1\big)'=(q_0-q_1)\,y_0 y_1.$$
\end{lemma}
\begin{proof}
Differentiating $y_0 y_1'-y_0' y_1$ and using $y_j''=-q_j y_j$ gives,
$$(y_0 y_1'-y_0' y_1)'=y_0' y_1'-y_0'' y_1 - y_0' y_1'+y_0 y_1'' = -(-q_0 y_0)y_1 + y_0(-q_1 y_1) = (q_0-q_1)y_0 y_1.$$

\end{proof}

\begin{theorem}\label{thm:sturm}
Let $q_1\ge q_0\ge 0$ on $[\alpha,\beta]$. If a nontrivial solution $y_0$ of $y_0''+q_0 y_0=0$ has consecutive zeros at $\xi<\eta$, then any solution $y_1$ of $y_1''+q_1 y_1=0$ has a zero in $(\xi,\eta)$.
\end{theorem}
\begin{proof}
Assume for contradiction that a solution $y_1$ has no zero in $(\xi,\eta)$. Multiply by $-1$ if necessary so that $y_1>0$ there. Between consecutive zeros, $y_0$ has constant sign. Replace $y_0$ by $-y_0$ if needed so $y_0>0$ on $(\xi,\eta)$. Define the Wronskian-like quantity $W:=y_0 y_1'-y_0' y_1$. By \cref{lem:lagrange}, $W'=(q_0-q_1)y_0 y_1\le 0$, so $W$ is nonincreasing. We then evaluate one-sided limits at the endpoints. Since $y_0(\xi)=0$ and $y_0'(\xi)>0$ (a simple zero with sign change from negative to positive), we have $W(\xi)=-y_0'(\xi)y_1(\xi)<0$. Similarly, $y_0(\eta)=0$ and $y_0'(\eta)<0$, hence $W(\eta)=-y_0'(\eta)y_1(\eta)>0$. But a nonincreasing $W$ cannot go from negative to positive, which proves the claim.

\end{proof}

\begin{theorem}\label{thm:separation}
If $y$ is a nontrivial solution of $y''+q y=0$ with continuous $q$, then between two consecutive zeros of $y$, every linearly independent solution has exactly one zero.
\end{theorem}
\begin{proof}
Fix consecutive zeros $\xi<\eta$ of $y$. By \cref{thm:sturm} applied with $q_0=q_1=q$, every solution has at least one zero in $(\xi,\eta)$. If a linearly independent solution $\tilde y$ had two or more zeros there, then the Wronskian $W=y\tilde y'-y'\tilde y$ would vanish at two points and hence be identically zero (for instance by Abel’s identity for continuous $q$), implying linear dependence. Thus exactly one zero occurs.

\end{proof}

\subsection{Prüfer transform and integral criterion}

\begin{definition}[Prüfer coordinates]\label{def:pruefer}
Assume $q\in C^1([a,b])$ and $q>0$ on $[a,b]$. For any nontrivial solution $y$, define functions $\rho>0$ and a continuous angle $\vartheta$ by,
$$y\coloneqq\rho\sin\vartheta,\qquad \frac{y'}{\sqrt{q}}\coloneqq \rho\cos\vartheta,$$
with the branch fixed by $\vartheta(a)=0$ if $y(a)=0$ and $y'(a)>0$.
\end{definition}

\begin{lemma}\label{lem:pruefer-ode}
With $q\in C^1$ and $q>0$, one has,
$$\vartheta'=\sqrt{q}+\frac{q'}{2q}\,\sin\vartheta\cos\vartheta,\qquad\frac{\rho'}{\rho}=-\frac{q'}{2q}\,\sin^2\vartheta.$$
\end{lemma}
\begin{proof}
Differentiating $y=\rho\sin\vartheta$ gives,
$$y'=\rho'\sin\vartheta+\rho\vartheta'\cos\vartheta.$$
But $y'=\sqrt{q}\,\rho\cos\vartheta$ by definition, so,
\begin{equation}\label{1} \rho'\sin\vartheta+\rho\vartheta'\cos\vartheta=\sqrt{q}\,\rho\cos\vartheta.
\end{equation}
Differentiating $y'/\sqrt{q}=\rho\cos\vartheta$ gives,
$$\frac{y''}{\sqrt{q}}-\frac{q'}{2q\sqrt{q}}\,y'=\rho'\cos\vartheta-\rho\vartheta'\sin\vartheta.$$
Since $y''=-q y=-q\rho\sin\vartheta$ and $y'=\sqrt{q}\,\rho\cos\vartheta$, this becomes,
\begin{equation}\label{2}
    -\sqrt{q}\,\rho\sin\vartheta - \frac{q'}{2q}\,\rho\cos\vartheta = \rho'\cos\vartheta-\rho\vartheta'\sin\vartheta
\end{equation}
Treating equations \eqref{1} and \eqref{2} as a linear system in the unknowns $\rho'$ and $\vartheta'$, we multiply \eqref{1} by $\sin\vartheta$ and \eqref{2} by $\cos\vartheta$, and add them to eliminate $\rho'$, giving,
$$\rho\vartheta'(\cos^2\vartheta+\sin^2\vartheta)= \sqrt{q}\,\rho\cos^2\vartheta - \sqrt{q}\,\rho\sin^2\vartheta - \frac{q'}{2q}\,\rho\cos^2\vartheta.$$
It follows that,
$$\vartheta'=\sqrt{q}\,(\cos^2\vartheta-\sin^2\vartheta)-\frac{q'}{2q}\cos^2\vartheta=\sqrt{q}\,\cos(2\vartheta)-\frac{q'}{2q}\frac{1+\cos(2\vartheta)}{2}.$$
Equivalently, using $\cos(2\vartheta)=1-2\sin^2\vartheta$ and $\sin\vartheta\cos\vartheta=\sin(2\vartheta)/2$, one simplifies to,
$$\vartheta'=\sqrt{q}+\frac{q'}{2q}\sin\vartheta\cos\vartheta.$$
To obtain $\rho'/\rho$, we multiply \eqref{1} by $\cos\vartheta$ and \eqref{2} by $\sin\vartheta$ and subtract them, eliminating $\theta'$. This gives,
$$\rho'\sin\vartheta\cos\vartheta+\rho\vartheta'\cos^2\vartheta= \sqrt{q}\,\rho\cos^2\vartheta,$$
$$-\sqrt{q}\,\rho\sin^2\vartheta - \frac{q'}{2q}\,\rho\sin\vartheta\cos\vartheta= \rho'\cos^2\vartheta-\rho\vartheta'\sin\vartheta\cos\vartheta.$$
Adding the two equations to eliminate $\vartheta'$ and solving for $\rho'/\rho$ yields,
$$\frac{\rho'}{\rho}=-\frac{q'}{2q}\sin^2\vartheta.$$
This completes the proof.

\end{proof}

\begin{lemma}\label{lem:theta-lb}
For $q\in C^1$ with $q>0$ and $a<b$,
\[
\vartheta(b)-\vartheta(a)\ \ge\ \int_a^b \sqrt{q}\,\dd\lam - \frac{1}{4}\int_a^b \frac{|q'|}{q}\,\dd\lam.
\]
\end{lemma}
\begin{proof}
From \cref{lem:pruefer-ode}, one has $\vartheta'=\sqrt{q}+(q'\sin\vartheta\cos\vartheta)/2q$. Using $|\sin\vartheta\cos\vartheta|\le 1/2$, one has,
$$\vartheta'\ \ge\ \sqrt{q}-\frac{|q'|}{4q}.$$
Integrating from $a$ to $b$ completes the proof.

\end{proof}

\begin{theorem}\label{thm:intcrit}
Let $q=f/m$ be continuous and nonnegative on $[a,b]$. If there exists a compact subinterval $[c,d]\subset[a,b]$ with
$$\int_c^d \sqrt{q(\lam)}\,\dd\lam \;=\;\int_c^d \sqrt{\frac{f(\lam)}{m}}\,\dd\lam\ >\ \pi,$$
then $y''+q y=0$ is not disconjugate on $[c,d]$. In particular, the solution with $y(c)=0$, $y'(c)=1$ has a second zero in $(c,d]$.
\end{theorem}
\begin{proof}
First we treat the case $q\in C^1$ and $q>0$ on $[c,d]$. Let $y$ solve $y''+q y=0$ with $y(c)=0$, $y'(c)>0$. In Prüfer coordinates, $\vartheta(c)=0$. By \cref{lem:theta-lb}, one has,
$$\vartheta(d)-\vartheta(c)\ \ge\ \int_c^d\sqrt{q}\,\dd\lam-\frac{1}{4}\int_c^d\frac{|q'|}{q}\,\dd\lam.$$
Since $\int_c^d \sqrt{q}>\pi$, continuity ensures a subinterval where the total lower bound still exceeds $\pi$ (shrink $[c,d]$ if necessary so the variation term is finite and small). When $\vartheta$ reaches $\pi$, $\sin\vartheta=0$ and hence $y=0$. Hence a zero occurs in $(c,d]$. For merely continuous $q\ge 0$, we approximate from below by $q_\epsilon\in C^1$ with $0<q_\epsilon\le q$ and $q_\epsilon\to q$ uniformly. Then $\int_c^d\sqrt{q_\epsilon}\to\int_c^d\sqrt{q}>\pi$. For each $\epsilon$, the equation $y_\epsilon''+q_\epsilon y_\epsilon=0$ has a second zero in $(c,d]$. By Sturm comparison (\cref{thm:sturm}), any solution for $q$ must have at least as many zeros. In particular the solution with $y(c)=0$, $y'(c)=1$ also has a second zero in $(c,d]$.

\end{proof}

\begin{corollary}\label{cor:const}
If $f\ge f_0>0$ on $[\lamo,\lamo+L]$, then every nontrivial solution of $u''+(f/m)u=0$ has a zero before
$$\lamo+\pi\,\sqrt{\frac{m}{f_0}}.$$
\end{corollary}
\begin{proof}
Here $q\equiv f_0/m$ is constant. The comparison equation $y''+(f_0/m)y=0$ has zeros spaced by $\pi\sqrt{m/f_0}$. Applying \cref{thm:sturm} with $q_1=q$ and $q_0=f_0/m$ yields the claim.

\end{proof}

\section{Conjugate points and blow-up}\label{sec:conjugate}

\begin{proposition}\label{prop:zero-conjugate}
Let $D(\lam)$ be the Jacobi map and $\mathcal A=|\det D|$. Then $u(\lam):=\mathcal A^{1/m}$ satisfies $u(\lam_*)=0$ if and only if $\det D(\lam_*)=0$, i.e., there exists a nontrivial screen Jacobi field vanishing at $\lam_*$ and $\gamma(\lam_*)$ is conjugate to $S$.
\end{proposition}
\begin{proof}
Since $\mathcal A=|\det D|$, we have $u(\lam_*)=0\iff\det D(\lam_*)=0$. If $\det D(\lam_*)=0$, the columns $J_1(\lam_*),\dots,J_m(\lam_*)$ of $D(\lam_*)$ are linearly dependent in $S_{\lam_*}$. Hence there exists a nontrivial combination $J=\sum_i \alpha_i J_i$ with $J(\lam_*)=0$. At $\lamo$, the set $\{J_i(\lamo)\}$ constitutes a basis of $\mathcal S_{\lamo}\cong T_{\gamma(\lamo)} S$, and so $J(\lamo)\in T S$ is non-zero. Thus $\gamma(\lam_*)$ is conjugate to $S$.

\end{proof}

\begin{proposition}\label{prop:blowup}
Assume $u>0$ on $(\lamo,\lam_*)$ and $u(\lam_*)=0$. Then $\theta=m\,u'/u\to -\infty$ as $\lam\to\lam_*$. Moreover, if $\theta(\lamo)<0$ then for $\lam\ge \lamo$ with $\theta<0$ one has,
$$
\theta(\lam)\le \frac{\theta(\lamo)}{1+\displaystyle\frac{1}{m}\theta(\lamo)(\lam-\lamo)}.
$$
Hence $\theta$ diverges to $-\infty$ by $\lam\le \lamo-m/\theta(\lamo)$ even if $f\equiv 0$.
\end{proposition}
\begin{proof}
As $\lam\to\lam_*$, one has $u(\lam)\to 0$ with $u'$ bounded (since solutions of linear ODEs are $C^2$), hence $u'/u\to -\infty$ from the left of the first zero. For the differential inequality, \cref{prop:ray} with $\|\sigma\|^2\ge 0$ and $\mathrm{Ric}(k,k)\ge 0$ gives,
$$\theta'\le -\frac{1}{m}\theta^2.$$
On any interval where $\theta<0$, this implies,
$$\left(\frac{1}{\theta}\right)'=-\frac{\theta'}{\theta^2}\ge \frac{1}{m}.$$
Integrating from $\lamo$ to $\lam$ gives $1/\theta(\lam)\ge 1/\theta(\lamo)+(\lam-\lamo)/m$, which can be rearranged to obtain the claimed bound.

\end{proof}

\section{Apparent horizons and failure of global visibility}\label{sec:visibility}
We are now finally in a position to relate all our analysis upto this point with the global visibility of singularities in a suitable class of spacetimes. To begin with, we recall a standard fact from Lorentzian geometry which states that null geodesics are not maximizing beyond their first focal point.
\begin{lemma}\label{standard}
    Let $B\subseteq M$ be a smooth codimension 2 spacelike submanifold of $M$ and $\gamma$ a null geodesic orthogonal to $B$. If $q=\gamma(\lambda_*)$ is focal to $B$ along $\gamma$, then $q\in I^{+}(B)$, and in particular $\gamma$ cannot be contained in $\partial J^+(B)=J^+(B)\setminus I^+(B)$ beyond $\lambda_*$. 
\end{lemma}
\begin{proof}
    See for instance the proof of Proposition 48 of Chapter 10 in \cite{oneill1983semiriemannian}.
    
\end{proof}

\begin{definition}[Locally and globally naked singularities]\label{eq:naked-sing}
    Let $(M,g)$ be an asymptotically flat spacetime with a conformal extension $(\tilde M,\tilde g)$, i.e., $\iota:M\hookrightarrow\tilde M$ is an isometric embedding, and let $\mathscr I^{+}$ denote the future null infinity of $M$. Let $\mathcal S \subseteq \partial_\iota(M)$ be a singular boundary set and consider a nested\footnote{This is a nested family in the sense that for $\varepsilon_1<\varepsilon_2$, one has $\mathcal S_{\varepsilon_1}\subseteq J^-(S_{\varepsilon_2})$.} family $\{\mathcal{S}_{\varepsilon}\}_{\varepsilon\in(0,\varepsilon_0]}$ of smooth codimension $2$ compact spacelike submanifolds of $M$ with $\overline{\iota(\mathcal S_{\varepsilon})}\to \mathcal S$ as $\varepsilon\to 0$ in the Kuratowski sense \cite{Kuratowski1966TopologyVol1}, i.e, the following conditions hold:
\begin{enumerate}
    \item For every open neighbourhood $U$ of $\mathcal{S}$ in $\tilde M$, there exists $\varepsilon_U>0$ such that $0<\varepsilon<\varepsilon_U\implies \iota(\mathcal{S}_{\varepsilon})\subset U$, and,
    \item For every $p\in \mathcal S$, and every open neighbourhood $V\ni p$ in $\tilde M$, there exists $\varepsilon^p_V>0$ such that $0<\varepsilon<\varepsilon^p_V\implies\iota(\mathcal{S}_\varepsilon)\cap V \neq \emptyset$.
\end{enumerate}
We say that the singular boundary $\mathcal{S}$ is \emph{globally visible} if there exists a sequence $\varepsilon_n\to 0$ with $J^+(\mathcal{S}_{\varepsilon_n})\cap \mathscr I^{+}\neq \emptyset$. Conversely, the singular boundary is not globally visible if there exists $\varepsilon_1>0$ such that for all $\varepsilon\in(0,\varepsilon_1]$, we have $J^{+}(\mathcal{S}_{\varepsilon})\cap \mathscr I^{+}=\emptyset$.
\end{definition}
 
A similar classification for local and global visibility was given by Eardley and Smarr \cite{Eardley:1978tr} in 1978. For instance, if the spacetime admits a continuous \emph{radius function} $r:M\to (0,\infty)$, such that the singular boundary set corresponds to $\{r=0\}$ in an extension, then one can take $\mathcal{S}_\varepsilon\coloneqq\{r=\varepsilon\}\cap\Sigma$ for a suitable spacelike hypersurface $\Sigma$ gives a family approaching $\{r=0\}$ in the above sense. See Figure \ref{fig-penrose} above for a demonstration of this definition in the case of the singularities arising from the spherically symmetric gravitational collapse of type-I matter fields \cite{Dwivedi1994OnTO}.

\begin{figure}[t]
\includegraphics[scale=0.8]{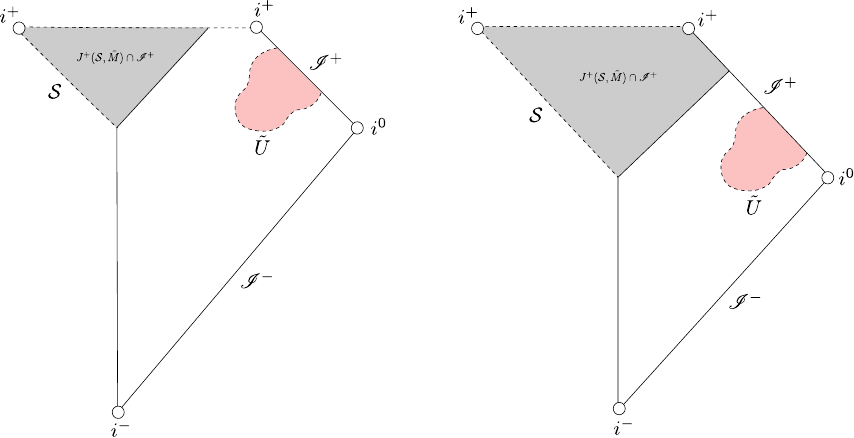}
\caption{Conformal extensions for the spherically symmetric spacetimes arising from the gravitational collapse of type-I matter fields. The extension on the left admits a locally visible singularity, whereas the extension on the right admits a globally visible singularity. The interior of the collapsing matter could be matched with a Schwarzschild exterior which is globally hyperbolic, and hence in particular $\tilde U$ is a causally simple neighbourhood of $\mathscr I^+$.} \label{fig-penrose}
\end{figure}
\begin{remark}
    Since $\mathscr{I}^+ \subset \tilde M\setminus M$, we may interpret $J^{+}(\mathcal{S}_{\varepsilon})\cap\mathscr I^{+}$ as $J^{+}_{\tilde M}(\iota(\mathcal S)_{\varepsilon})\cap \mathscr{I}^{+}$, where $J^{+}_{\tilde M}$ denotes the causal future in $(\tilde M, \tilde g)$
\end{remark}

\begin{lemma}\label{lem:weak_to_prompt}
Assume that the future null infinity $\mathscr I^+$ admits a continuous function
$u:\mathscr I^+\to\R$ such that,
\begin{enumerate}
\item for every compact interval $[a,b]\subset\R$ the set
$u^{-1}([a,b])\subset\mathscr I^+$ is compact,
\item  for every $x\in\mathscr I^+$ and every open neighborhood
$W\ni x$ in $\mathscr I^+$ there exists $y\in W$ with $u(y)<u(x)$, and
\item the causal relation
$J^+_{\tilde M}\subset\tilde M\times\tilde M$ is closed on a neighborhood of
$\iota(M)\cup\mathscr I^+$ (equivalently, $(\tilde M,\tilde g)$ is causally simple there).
\end{enumerate}
Let $K\subset M$ be compact. If $J^+_{\tilde M}(\iota(K))\cap\mathscr I^+\neq \emptyset$ then $\partial J^+_{\tilde M}(\iota(K))\cap\mathscr I^+\neq\emptyset$.
\end{lemma}
\begin{proof}
Set $A:=J^+_{\tilde M}(\iota(K))\cap\mathscr I^+$ and assume $A\neq\emptyset$. Let $u_0:=\inf\{u(x):x\in A\}.$
Choose a sequence $x_n\in A$ with $u(x_n)\to u_0$. By (1), we may extract a subsequence with $x_n\to x_0\in\mathscr I^+$ and $u(x_0)=u_0$. We claim $x_0\in J^+_{\tilde M}(\iota(K))$. Indeed, pick $p_n\in \iota(K)$ with $p_n\le_{\tilde M} x_n$.
Since $\iota(K)$ is compact, we may pass to $p_n\to p_0\in\iota(K)$. By (3) (i.e., by closedness of $J^+_{\tilde M}$), we have
$(p_0,x_0)\in J^+_{\tilde M}$, and hence $x_0\in J^+_{\tilde M}(\iota(K))$, i.e.\ $x_0\in A$. If $x_0\in\mathrm{Int}\big(J^+_{\tilde M}(\iota(K))\big)$, we choose an open set $U\subset\tilde M$ with
$x_0\in U\subset J^+_{\tilde M}(\iota(K))$. Then $W:=U\cap\mathscr I^+$ is an open neighborhood of $x_0$ in $\mathscr I^+$.
By (2), we may choose $y\in W$ with $u(y)<u(x_0)=u_0$. But $y\in U\subset J^+_{\tilde M}(\iota(K))$, and hence $y\in A$,
contradicting the definition of $u_0$. Therefore $x_0\notin\mathrm{Int}\big(J^+_{\tilde M}(\iota(K))\big)$, and so
$x_0\in\partial J^+_{\tilde M}(\iota(K))\cap\mathscr I^+$.

\end{proof}

\begin{corollary}\label{rem:weak_strong_equiv}
Under the hypotheses of Lemma~\ref{lem:weak_to_prompt}, for every compact $K\subset M$ one has
$$
J^+_{\tilde M}(\iota(K))\cap\mathscr I^+\neq\emptyset
\quad\Longleftrightarrow\quad
\partial J^+_{\tilde M}(\iota(K))\cap\mathscr I^+\neq\emptyset.
$$
\end{corollary}
\begin{proof}
($\Rightarrow$) This is Lemma~\ref{lem:weak_to_prompt}.

($\Leftarrow$) We first show that $J^+_{\tilde M}(\iota(K))$ is closed in $\tilde M$. Let $x_n\in J^+_{\tilde M}(\iota(K))$ with $x_n\to x$. Choose $p_n\in \iota(K)$ such that $p_n\le_{\tilde M} x_n$. Since $\iota(K)$ is compact, we may pass to a subsequence $p_{n_j}\to p\in \iota(K)$. By point (3) in Lemma~\ref{lem:weak_to_prompt}, the set $J^+_{\tilde M}=\{(a,b): a\le_{\tilde M} b\}$ is closed in $\tilde M\times\tilde M$, hence from $(p_{n_j},x_{n_j})\to(p,x)$ we obtain $p\le_{\tilde M} x$. Thus $x\in J^+_{\tilde M}(\iota(K))$. Therefore $J^+_{\tilde M}(\iota(K))$ is closed. Now, since $\partial A\subset \overline{A}$ for every set $A$, and here $\overline{J^+_{\tilde M}(\iota(K))}=J^+_{\tilde M}(\iota(K))$,
we get $\partial J^+_{\tilde M}(\iota(K))\subset J^+_{\tilde M}(\iota(K)).$ Hence any point of $\partial J^+_{\tilde M}(\iota(K))\cap\mathscr I^+$ lies in $J^+_{\tilde M}(\iota(K))\cap\mathscr I^+$.

\end{proof}

\begin{lemma}\label{lem:generators}
Let $B\subset M$ be a smooth compact spacelike codimension $2$ submanifold. If
$x\in J^+_{\tilde M}(\iota(B))\setminus\iota(B)$, then there exists a future-directed null geodesic segment
$\gamma:[0,1]\to\tilde M$ such that
$$
\gamma(0)\in\iota(B),\qquad \gamma(1)=x,\qquad \gamma([0,1])\subset \partial J^+_{\tilde M}(\iota(B)).
$$
Moreover, $k(0)\perp T_{\gamma(0)}\iota(B)$.
\end{lemma}
\begin{proof}
This is the standard structure theorem for achronal boundaries. See for instance the proof of Corollary 5 of Chapter 14 in \cite{oneill1983semiriemannian}.

\end{proof}

\begin{theorem}\label{thm:global_visibility_failure}
Assume the hypotheses of Lemma~\ref{lem:weak_to_prompt} hold. Let $S\subset\partial_\iota(M)$ be a singular boundary set and
$\{S_\varepsilon\}_{\varepsilon\in(0,\varepsilon_0]}$ a nested family of smooth compact spacelike codimension $2$ submanifolds of $M$ with $\iota(S_\varepsilon)\to S$ in the sense of Definition \ref{eq:naked-sing}. Fix $\varepsilon\in(0,\varepsilon_0]$. Assume that every null generator $\gamma$ of $\partial J^+_{\tilde M}(\iota(S_\varepsilon))$ which reaches $\mathscr I^+$ is contained
in a $C^2$ null hypersurface generated by null geodesics orthogonal to $S_\varepsilon$, so that $\omega\equiv0$.  Moreover assume that the Jacobi map $D$ and $u:=|\det D|^{1/m}$, $m:=N-1$, are defined up to the first focal point), and that along $\gamma$ the scalar ODE,
$$
u''+\frac{1}{m}f\,u=0,\qquad f=\|\sigma\|^2+\mathrm{Ric}(k,k)\ge0,
$$
holds on intervals where $u>0$.
Assume in addition that for every such generator $\gamma$ there exists a subinterval $[c,d]$ of its affine parameter domain with $u(c)\neq 0$ and
\begin{equation}\label{eq:prufer_global}
\int_c^d \sqrt{\frac{f(\lambda)}{m}}\,d\lambda>\pi.
\end{equation}
Then $J^+_{\tilde M}(\iota(S_\varepsilon))\cap\mathscr I^+=\emptyset.$ Consequently, if \eqref{eq:prufer_global} holds for all $\varepsilon\in(0,\varepsilon_1]$ for some $\varepsilon_1>0$, then the singular boundary set
$S$ is \emph{not globally visible}, i.e.\ there exists $\varepsilon_1>0$ such that for all
$\varepsilon\in(0,\varepsilon_1]$ one has $J^+_{\tilde M}(\iota(S_\varepsilon))\cap\mathscr I^+=\emptyset$.
\end{theorem}
\begin{proof}
Fixing $\varepsilon\in(0,\varphi_0]$, we argue by contradiction. Suppose $J^+_{\tilde M}(\iota(S_\varepsilon))\cap\mathscr I^+\neq\emptyset$. By
Lemma~\ref{lem:weak_to_prompt} there exists $x\in \partial J^+_{\tilde M}(\iota(S_\varepsilon))\cap\mathscr I^+.$ By Lemma~\ref{lem:generators}, $x$ lies on a null generator $\gamma\subset \partial J^+_{\tilde M}(\iota(S_\varepsilon))$
emanating orthogonally from $\iota(S_\varepsilon)$ and reaching $x$. Applying the assumption \eqref{eq:prufer_global} to this $\gamma$, one obtains (via Theorem \ref{thm:intcrit}) a zero $\lambda_* \in (c,d]$ of $u$ along $\gamma$. By Proposition \ref{prop:zero-conjugate}, $\gamma(\lambda_*)$ is focal to $B:=\iota(S_\varepsilon)$ along $\gamma$. Lemma \ref{standard} now implies $\gamma(\lambda)\in I^+_{\tilde M}(B)$ for all $\lambda>\lambda_*$, hence $\gamma$ cannot be contained in $\partial J^+_{\tilde M}(B)$ beyond $\lambda_*$. This contradicts that $\gamma\subset\partial J^+_{\tilde M}(B)$ reaches $x\in\mathscr I^+$. Therefore $J^+_{\tilde M}(\iota(S_\varepsilon))\cap\mathscr I^+=\emptyset$. The final statement follows immediately from the definition of global visibility.

\end{proof}

\begin{remark}
    One could also formulate a version of Theorem \ref{thm:global_visibility_failure} that uses the barrier criterion given by Theorem \ref{thm:barrier}. However for application purposes, the Sturm criterion is much more practical to use, as will be demonstrated in case of the Einstein massless-scalar field system.
\end{remark}

\section{Application to gravitational collapse of a scalar field.}

We work in $3+1$ dimensions. Let $( M,g,\phi:M\to \mathbb{R})$ be a \emph{smooth} spherically symmetric solution of the Einstein-massless scalar field system, in the sense of Christodoulou \cite{Christodoulou1994,Christodoulou1999}. We only use local analysis on the
regular region, and we assume that the metric $g$ is at least $C^2$.

\begin{definition}[Spherical symmetry]\label{def:spherical}
A spacetime $(M,g)$ is \emph{spherically symmetric} if there is an effective isometric
action of $\mathrm{SO}(3)$ on $M$ whose orbits are spacelike $2$-spheres on an open dense set.
\end{definition}

\begin{lemma}\label{lem:warped}
Let $(M,g)$ be spherically symmetric and let $p\in M$ lie on a principal orbit $\mathcal O_p\simeq \mathbb S^2$. Then there is a neighborhood $U\ni p$ and a $2$-dimensional Lorentzian manifold $(\mathcal Q,h)$ with coordinates $(x^0,x^1)$ on a chart $V\subset\mathcal Q$, together with a smooth function $r:V\to(0,\infty)$, such that
\begin{equation}\label{eq:warped}
(U,g)\ \cong\ (V\times \mathbb S^2, h_{ab}(x)\,\dd x^a \otimes \dd x^b + r(x)^2\,(g_{\mathbb{S}^2})_{AB}(\omega)\ \dd\omega^A \otimes  \dd\omega^B), 
\end{equation}
where $g_{\mathbb{S}^2}$ is the standard round metric on $\mathbb S^2$.
\end{lemma}
\begin{proof}
Fix $p$ on a principal orbit $\mathcal O_p$. Let $E_p:=T_p(\mathcal O_p)$ (a spacelike $2$-plane). By the $\mathrm{SO}(3)$-invariance of $g$, the orthogonal complement $E_p^\perp$ is invariant under the isotropy subgroup at $p$, and is $2$-dimensional with Lorentzian signature (since $\dim M=4$ and $\mathcal O_p$ is spacelike). By the slice theorem for proper Lie group actions, there is a local $\mathrm{SO}(3)$-equivariant diffeomorphism of a neighborhood of $p$ with a product $V\times \mathbb S^2$
where $V$ is a neighborhood in the orbit space $\mathcal Q:=M/\mathrm{SO}(3)$, and the metric splits as \eqref{eq:warped} with some Lorentzian metric $h$ on $V$ and a warping factor $r^2$ multiplying the unique $\mathrm{SO}(3)$-invariant metric on each orbit (which must be a round metric up to scale). The scale is encoded by the smooth function $r>0$ (the \emph{area radius}), defined for $x\in V$ by,
$$\mathrm{Area}(\{x\}\times \mathbb S^2)=4\pi r(x)^2.$$
This completes the proof.

\end{proof}

\begin{lemma}\label{lem:double-null}
Let $(\mathcal Q,h)$ be any $2$-dimensional Lorentzian manifold. Then there exist coordinates $(u,v)$ on $\mathcal Q$ such that,
\begin{equation}\label{eq:2d-conformal}
h = -\Omega(u,v)^2 \dd u\otimes \dd v
\end{equation}
for some smooth $\Omega>0$. Consequently, on $U\simeq V\times \mathbb S^2$ the spacetime metric reads,
\begin{equation}\label{eq:double-null}
g = -\Omega(u,v)^2\ \dd u \otimes \dd v + r(u,v)^2\,(g_{\mathbb{S}^2})_{AB}(\omega)\ \dd\omega^A\otimes \dd\omega^B.
\end{equation}
\end{lemma}
\begin{proof}
In dimension $2$, every Lorentzian metric is locally conformally flat: choose local null vector fields
$L,\underline L$ with $h(L,L)=h(\underline L,\underline L)=0$ and $h(L,\underline L)<0$. Let $u$ (resp. $v$)
solve $L(u)=0$ (resp. $\underline L(v)=0$) with non-vanishing gradients. Then $\dd u$ and $\dd v$ are null covectors,
and $h$ must be proportional to $\dd u\otimes \dd v$. Writing the proportionality factor as $-\Omega^2$ gives \eqref{eq:2d-conformal}. Substituting into \eqref{eq:warped} yields \eqref{eq:double-null}.

\end{proof}

Now fix a point $(u_*,v_*,\omega_*)$ with $r(u_*,v_*)>0$ and consider the outgoing null curve $\gamma:\mathbb{R}\to M$ given by,
$$ \gamma(v):=(u_*,v,\omega_*).$$ Its tangent is $\partial_v$, which is null since $g_{vv}=0$ in \eqref{eq:double-null}. We now choose an affine parametrization.

\begin{lemma}\label{lem:affine}
Along each outgoing null curve $u=\mathrm{const}$, the vector field
\begin{equation}\label{eq:k}
k := \Omega^{-2}\partial_v
\end{equation}
satisfies $\nabla_k k=0$. Hence $k$ is tangent to an affinely parametrized null geodesic.
\end{lemma}
\begin{proof}
Since $\omega$ is constant on $\gamma$, only the $(u,v)$-part matters. In coordinates $(u,v)$ on $\mathcal Q$,
the only nonzero component of $h$ is $h_{uv}=h_{vu}=-\Omega^2/2$ (up to the conventional symmetrization).
A direct Christoffel computation for $h=-\Omega^2 \dd u\otimes \dd v$ gives,
$$\nabla_{\partial_v}\partial_v = 2(\partial_v\log\Omega)\,\partial_v,$$
i.e. $\partial_v$ is geodesic but not affine unless $\partial_v\Omega=0$. Rescaling by $\Omega^{-2}$ removes the
non-affinity, giving,
$$
\nabla_k k = \nabla_{\Omega^{-2}\partial_v}(\Omega^{-2}\partial_v) = \Omega^{-4}\nabla_{\partial_v}\partial_v + \Omega^{-2}(\partial_v\Omega^{-2})\,\partial_v = \Omega^{-4}\cdot 2(\partial_v\log\Omega)\partial_v -2\Omega^{-4}(\partial_v\log\Omega)\partial_v =0.
$$
The $\mathbb S^2$-part contributes nothing because $k$ has no angular components and the metric is a warped product. This completes the proof.

\end{proof}

\begin{definition}\label{def:screen}
At a point $p$ on a null geodesic $\gamma$ with tangent $k$, define the \emph{screen space}
$$\mathcal S_p := k_p^\perp / \mathrm{span}\{k_p\}.$$
Along $\gamma$, the disjoint union $\mathcal S:=\bigsqcup_{p\in\gamma}\mathcal S_p$ is the \emph{screen bundle}.
Its rank is $m:=\dim\mathcal S_p=2$ in $3+1$ dimensions.
\end{definition}

\begin{lemma}\label{lem:screen=TS2}
Let $p=(u,v,\omega)\in M$ with $r(u,v)>0$ and let $k=\Omega^{-2}\partial_v$ at $p$.
Then the map
$$
T_\omega\mathbb S^2 \longrightarrow \mathcal S_p,\qquad X \mapsto [\tilde X],
$$
where $\tilde X$ is the corresponding angular tangent vector in $T_pM$ (i.e.\ $\tilde X$ has only $\omega$-components), is a linear isomorphism. In particular, $\mathcal S_p$ is naturally identified with $T_\omega\mathbb S^2$.
\end{lemma}

\begin{proof}
Using \eqref{eq:double-null}, the tangent space may be decomposed as,
$$T_pM \cong T_{(u,v)}\mathcal Q \oplus T_\omega\mathbb S^2,$$
This decomposition is orthogonal with respect to $g$, with $g|_{T_\omega\mathbb S^2}=r(u,v)^2g_{\mathbb{S}^2}$ positive definite. Since $k\in T_{(u,v)}\mathcal Q$, for any $\tilde X\in T_\omega\mathbb S^2$ we have $g(k,\tilde X)=0$, and hence $\tilde X\in k^\perp$. Moreover $\tilde X\notin \mathrm{span}\{k\}$ unless $\tilde X=0$ (since they live in orthogonal summands). Hence the class $[\tilde X]\in \mathcal S_p$ is well-defined and is injective in $X$. Conversely, take any $Y\in k^\perp\subset T_pM$ and decompose $Y=Y_{\mathcal Q}+Y_{\mathbb S^2}$. Then $0=g(k,Y)=g(k,Y_{\mathcal Q})$ (since $k\perp T_\omega\mathbb S^2$). In the $2$-dimensional Lorentzian space $T_{(u,v)}\mathcal Q$, the orthogonal complement of a null vector is $1$-dimensional and equals its span. Hence $Y_{\mathcal Q}\in \mathrm{span}\{k\}$. Therefore $[Y]=[Y_{\mathbb S^2}]$ in $k^\perp/\mathrm{span}\{k\}$, proving
surjectivity of $T_\omega\mathbb S^2\to\mathcal S_p$. This completes the proof.

\end{proof}

Let $S_{u,v}:=\{(u,v)\}\times \mathbb S^2$ denote the \emph{symmetry $2$-sphere} through $(u,v)$. By Lemma~\ref{lem:screen=TS2}, tangent vectors to $S_{u,v}$ represent screen directions.

\begin{definition}\label{def:null2ff}
Fix $p\in S_{u,v}$ and let $k$ be the affinely parametrized tangent to an outgoing null generator through $p$. For $X,Y\in T_pS_{u,v}$, we define the \emph{null second fundamental form} by
$$B(X,Y) := g(\nabla_X k, Y).$$
Since $g|_{T_pS_{u,v}}$ is positive definite, there exists a unique endomorphism $b:T_pS_{u,v}\to T_pS_{u,v}$ such that for all $X,Y\in T_pS_{u,v}$, we have $B(X,Y)=g(bX,Y)$. The \emph{null expansion} is $\theta:=\mathrm{tr}\,b$, and the \emph{shear} is the trace-free part $\hat b:=b-(\theta/2)\,\mathrm{id}$ (so $\sigma^2:=\mathrm{tr}(\hat b^2)\ge 0$).
\end{definition}

\begin{lemma}\label{lem:Bexplicit}
Along an outgoing radial null generator $\gamma$ with affine tangent $k$, one has,
\begin{equation}\label{eq:Bexplicit}
B = \frac{k(r)}{r}\,g\big|_{T S_{u,v}},
\qquad\text{i.e.}\qquad
b = \frac{k(r)}{r}\,\mathrm{id}_{T S_{u,v}}.
\end{equation}
Consequently,
\begin{equation}\label{eq:theta-shear}
\theta = 2\,\frac{k(r)}{r},
\qquad
\hat b \equiv 0,
\qquad
\sigma\equiv 0.
\end{equation}
\end{lemma}
\begin{proof}
Let $\omega^A$ be local coordinates on $\mathbb S^2$ and write $e_A:=\partial_{\omega^A}$ as coordinate vector fields tangent to $S_{u,v}$. Then $[e_A,\partial_v]=0$, and hence $[e_A,k]=0$ as well. Using metric compatibility and $g(k,e_A)=0$ yields,
$$ k\big(g(e_A,e_B)\big)= g(\nabla_k e_A,e_B)+g(e_A,\nabla_k e_B).$$
But $[e_A,k]=0$ implies $\nabla_k e_A-\nabla_{e_A}k=0$, and hence,
$$k\big(g(e_A,e_B)\big)= g(\nabla_{e_A}k,e_B)+g(e_A,\nabla_{e_B}k)=B(e_A,e_B)+B(e_B,e_A).$$
Thus $B$ is symmetric and moreover $2B(e_A,e_B)=k\big(g(e_A,e_B)\big).$ From \eqref{eq:double-null}, one has $g(e_A,e_B)=r(u,v)^2(g_{\mathbb{S}^2})_{AB}(\omega)$, and hence along $\gamma$ (where $u$ and $\omega$
are constant), one has,
$$ k\big(g(e_A,e_B)\big)=k(r^2)\,(g_{\mathbb{S}^2})_{AB}=2r\,k(r)\,(g_{\mathbb{S}^2})_{AB}.$$
Therefore,
$$B(e_A,e_B)= r\,k(r)\,(g_{\mathbb{S}^2})_{AB} = \frac{k(r)}{r}\,(r^2(g_{\mathbb{S}^2})_{AB})= \frac{k(r)}{r}\,g(e_A,e_B),$$
which is exactly \eqref{eq:Bexplicit}. Tracing with respect to the induced metric on $S_{u,v}$ yields $\theta = 2(k(r)/r)$, and since $b$ is a scalar multiple of the identity, its trace-free part vanishes.

\end{proof}

The Einstein-massless scalar field system is given by,
$$
G_{ab}= T_{ab},
\qquad
T_{ab}=\nabla_a\phi\,\nabla_b\phi - \frac12 g_{ab}(\nabla\phi)^2,
\qquad
\Box_g\phi=0.
$$

\begin{lemma}\label{lem:ricci-scalar}
In $3+1$ dimensions, the Einstein equation implies \begin{equation}\label{eq:ricci=gradphi}
\mathrm{Ric}_{ab} = \nabla_a\phi\,\nabla_b\phi.
\end{equation}
In particular, for any null vector $k$,
\begin{equation}\label{eq:Ric-kk}
\mathrm{Ric}(k,k)=(k\phi)^2.
\end{equation}
\end{lemma}
\begin{proof}
Taking the trace of $T_{ab}$ gives,
$$
T:=g^{ab}T_{ab} = (\nabla\phi)^2 - 2(\nabla\phi)^2 = -(\nabla\phi)^2.
$$
Tracing the Einstein equation $G_{ab}=T_{ab}$ gives,
$$
G:=g^{ab}G_{ab}= -R =  T = -(\nabla\phi)^2.
$$
Hence $R=(\nabla\phi)^2$. We rewrite $G_{ab}=R_{ab}-(R/2)g_{ab}=T_{ab}$ as
$$
R_{ab} = \nabla_a\phi\nabla_b\phi - \frac12 g_{ab}(\nabla\phi)^2 + \frac12 R g_{ab}.
$$
Substituting $R=(\nabla\phi)^2$ gives $R_{ab} = \nabla_a\phi\nabla_b\phi.$ Contracting with $k^ak^b$ and using $g(k,k)=0$ yields \eqref{eq:Ric-kk}.

\end{proof}

\begin{proposition}\label{prop:sturm-scalar}
Along any outgoing radial null generator $\gamma$ with affine tangent $k$ in a spherically symmetric
Einstein-scalar field spacetime,
\begin{equation}\label{eq:sturm-final}
u'' + \frac{1}{2}(k\phi)^2\,u = 0,
\end{equation}
Moreover, one has $u=r/r(\lambda_0)$, and so equivalently
\begin{equation}\label{eq:r-ode}
r'' + \frac{1}{2}(k\phi)^2\,r = 0.
\end{equation}
\end{proposition}
\begin{proof}
Spherical symmetry gives $\sigma\equiv 0$ by Lemma~\ref{lem:Bexplicit}. Then equation \eqref{eq:sturm-final} is simply the Raychaudhuri equation for the Einstein massless-scalar field system. Now, if $E_A(\lambda)$ ($A=1,2$) is a parallelly -transported
orthonormal basis of $S_{\lambda_0}$, then the corresponding screen Jacobi fields satisfy
$J_A(\lambda) = (r(\lambda)/r(\lambda_0) )\cdot E_A(\lambda)$, and so the Jacobi map is $D(\lambda)=(r(\lambda)/r(\lambda_0))\mathrm{Id}$ on $S_{\lambda_0}$. It follows that,
$$
A(\lambda)=|\det D(\lambda)| = \left(\frac{r(\lambda)}{r(\lambda_0)}\right)^2,
\qquad
u(\lambda)=A(\lambda)^{1/2}=\frac{r(\lambda)}{r(\lambda_0)}.
$$
In particular, $u$ is (up to a harmless constant factor) just the area radius $r$ along $\gamma$. This completes the proof.

\end{proof}

The Pr\"ufer/Sturm focusing criterion given by Theorem \ref{thm:intcrit} (specialized to $m=2$ and $\sigma=0$) says: if for some interval $[\lambda_c,\lambda_d]$, one has,
\begin{equation}\label{eq:prufer-paper}
\int_{\lambda_c}^{\lambda_d}\sqrt{q(\lambda)}\,d\lambda > \pi,
\qquad
q(\lambda):=\frac12\mathrm{Ric}(k,k)= \frac{1}{2} (k\phi)^2,
\end{equation}
then $u$ has a zero in $(\lambda_c,\lambda_d)$, hence $\gamma$ has a conjugate point in $(\lambda_c,\lambda_d)$. Substituting $q=(k\phi)^2/2$ gives
\begin{equation}\label{eq:flux-affine}
\int_{\lambda_c}^{\lambda_d} |k\phi|\dd\lambda > \sqrt{2}\pi
\end{equation}

\begin{lemma}\label{lem:flux-v}
Along an outgoing generator $u=u_*$ with $k=\Omega^{-2}\partial_v$ and affine parameter $\lambda$, one has $\dd\lambda=\Omega^2 \dd v$ and $(k\phi)\dd\lambda = (\partial_v\phi)\dd v$. Consequently
\begin{equation}\label{eq:flux-v}
\int_{\lambda_c}^{\lambda_d}|k\phi|\ \dd\lambda= \int_{v_c}^{v_d}|\partial_v\phi(u_*,v)|\ \dd v.
\end{equation}
\end{lemma}

\begin{proof}
Since $k\equiv \dd/\dd\lambda$ and $k=\Omega^{-2}\partial_v$, we have,
$$
\frac{\dd}{\dd\lambda} = \Omega^{-2}\frac{\partial}{\partial v}
\quad\Longrightarrow\quad
\frac{\dd v}{\dd\lambda}=\Omega^{-2}
\quad\Longrightarrow\quad
\dd\lambda=\Omega^2 \dd v.
$$
Then,
$$
(k\phi) \dd\lambda = (\Omega^{-2}\partial_v\phi)\cdot(\Omega^2\dd v)=(\partial_v\phi)\dd v,
$$
Taking absolute values and integrating yields \eqref{eq:flux-v}.
\end{proof}

\begin{proposition}\label{prop:explicit-threshold}
If along an outgoing radial null generator $u=u_*$ there exist $v_c<v_d$ in the regular region such that
\begin{equation}\label{eq:threshold}
\int_{v_c}^{v_d} |\partial_v\phi(u_*,v)|\ \dd v > \sqrt{2} \pi,
\end{equation}
then that generator contains a pair of null conjugate points (equivalently, $u=r/r(\lambda_0)$ vanishes) in $(v_c,v_d)$.
\end{proposition}
\begin{proof}
By Lemma~\ref{lem:flux-v}, equation \eqref{eq:threshold} is equivalent to the inequality in \eqref{eq:flux-affine}. That is exactly \eqref{eq:prufer-paper}. Applying Theorem \ref{thm:intcrit} to the equation \eqref{eq:sturm-final} gives a zero of $u$, and hence a conjugate point, in the interior.

\end{proof}

Christodoulou's scalar field collapse models are asymptotically flat \cite{Christodoulou:1987vv}, and a subclass of his models admit locally naked singularities. Moreover the geometry near future null infinity is shown to behave like vacuum (which in spherical symmetry is uniquely determined by the globally hyperbolic schwarzschild exterior spacetime. This means in particular that his models are asymptotically flat and globally hyperbolic (and hence causally simple) in a neighbourhood of $\mathscr{I}^+$. Moreover, there exists a \emph{retarded time function} $\mathsf u: \mathscr I^{+}\to \mathbb R$ which is proper and admits no local minumum on $\mathscr{I}^{+}$. It follows that Lemma \ref{lem:weak_to_prompt} holds for the Einstein-massless scalar field system. Therefore, if the lower bound given by equation \eqref{eq:threshold} holds along every outgoing null geodesic issuing from the locally naked singularity (in the specified subclass), then Theorem \ref{thm:global_visibility_failure} guarantees that these geodesics cannot escape to $\mathscr I^{+}$, and hence the singularity is not globally visible. Equation \eqref{eq:threshold} represents a lower bound on the integrated flux of the scalar field along null geodesics emanating from the singularity. Conversely, if the singularity is globally naked, then along all outgoing null geodesics emanating from the singularity, one has $\int_{\gamma} |\partial_v\phi|\ \dd \lambda<\sqrt{2}\pi$, i.e., the integrated scalar field flux admits a uniform upper bound.

\printbibliography

\end{document}